\documentclass[11pt]{article}
\usepackage{amssymb}
\usepackage{amsmath}
\usepackage{graphics}
\usepackage{color} 
\usepackage{fullpage}

\usepackage{vmargin}
\setmarginsrb{28mm}{10mm}{20mm}{15mm}{12pt}{11mm}{0pt}{11mm}

\newtheorem{theorem}{Theorem}
 \newtheorem{corollary}[theorem]{Corollary}
 \newtheorem{lemma}[theorem]{Lemma}
 \newtheorem{proposition}[theorem]{Proposition}
\newenvironment{proof}[1][Proof]{\noindent\textbf{#1.} }{\hfill $\Box$\\[2mm]} 
\newenvironment{proofsketch}[1][Proof sketch]{\noindent\textbf{#1.} }{\hfill $\Box$\\[2mm]} 
\def\lf{\tiny}

\def\nnll{\refstepcounter{linenumber}\lf\thelinenumber}
\newcounter{linenumber}

\def\Time{\mathbb{T}}

\def\D{\ensuremath{\mathcal{D}}}

\def\A{\ensuremath{\mathcal{A}}}
\def\B{\ensuremath{\mathcal{B}}}

\def\R{\ensuremath{\mathcal{R}}}

\def\I{\ensuremath{\mathcal{I}}}
\def\O{\ensuremath{\mathcal{O}}}

\def\E{\ensuremath{\mathcal{E}}}

\def\Nat{\ensuremath{\mathbb{N}}}

\def\fd{failure detector}

\def\cons{\textit{cons}}
\def\nbwf{\ensuremath{m}}

\def\fd{\ensuremath{S}}
\def\wf{\ensuremath{C}}
\def\nbfinal{\ensuremath{n}}

\newcommand{\correct}{\mathit{correct}}
\newcommand{\faulty}{\mathit{faulty}}
\newcommand{\infi}{\mathit{inf}}

\newcommand{\true}{\mathit{true}}
\newcommand{\false}{\mathit{false}}

\newcommand{\remove}[1]{}

\def\Nomega{\ensuremath{\neg\Omega}}
\def\Vomega{\ensuremath{\overrightarrow{\Omega}}}

\newcommand{\id}[1]{\mbox{\textit{#1}}}

\newcommand{\Leader}{\textit{Leader}}

\newcommand{\ignore}[1]{}

\begin{document}

\bibliographystyle{abbrv}


\title{Wait-Freedom with Advice} 

\author{
Carole Delporte-Gallet$^1$~~~~Hugues Fauconnier$^1$\\
\\
Eli Gafni$^2$~~~~Petr Kuznetsov$^3$
\\
\\
\\
\large $^1$ University Paris Diderot\\
\large $^2$ UCLA\\
\large $^3$TU Berlin/Deutsche Telekom Laboratories
}

\date{}

\maketitle

\begin{abstract}
We motivate and propose a new way of thinking about failure
detectors which allows us to define, quite surprisingly, what it means
to solve a distributed task 
\emph{wait-free}
\emph{using a failure detector}.
In our model, the system is composed of \emph{computation} processes that obtain
inputs and are supposed to output in a finite number of steps and
\emph{synchronization} processes that are subject to failures and can query a failure
detector.
We assume that, under the condition that  \emph{correct} synchronization processes take sufficiently many steps,   
they provide the computation processes with enough \emph{advice}  to solve the given task wait-free: 
every computation process outputs in a finite number of its own steps, regardless of the behavior of other
computation processes.
%
Every task can thus be characterized by the \emph{weakest} failure detector that
allows for solving it, and we show that every such
failure detector captures a form of set agreement.   
We then obtain a complete classification of tasks, 
including ones that evaded comprehensible characterization  
so far, such as renaming or weak symmetry breaking.
\end{abstract}

\thispagestyle{empty}

\clearpage

\pagenumbering{arabic}

\section{Introduction}

\paragraph{What does it mean to solve a task?} A distributed task for
a set of processes can be seen as a function 
that maps an input vector to an output vector, one value per process. 
It is easy to reason about correctness of a task solution by matching
the outputs to the inputs with respect to the task specification. 
When it comes to progress, however, it is getting less trivial.  

On the surface, it is desirable 
to expect that the input vector is exactly matched by the output
vector, i.e., every participating process obtains  an
output.\footnote{A process is considered participating if it takes at
  least one steps in the computation.}
Unfortunately, in asynchronous or partially synchronous systems where relative processes'
speeds are unbounded or very large, ensuring this property would
require very long waiting. 
A more natural \emph{wait-freedom} property requires that
any participating process that takes sufficiently many steps obtains
an output,  ``regardless of execution speeds of other processes''~\cite{Her91}.
A wait-free task solution thus allows for treating the requirement
``a given participant outputs'' as a \emph{liveness}
property~\cite{AS85}: every execution has an extension in which the
requirement is met.
Naturally, wait-freedom assumes no notion of process \emph{failures}: a process that
does not take steps for a while in a given execution, always has a
chance to wake up and take enough steps to output.
%

\paragraph{Failure detectors.} Unfortunately, very few tasks can be solved
wait-free in the basic read-write shared-memory model~\cite{FLP85,LA87,HS99,SZ00,BG93b,BGLR01}.
The \emph{failure detector} abstraction~\cite{CHT96,CT96} 
was proposed to circumvent these impossibilities. 
Intuitively, a failure detector provides each process with some (possibly 
incomplete and inaccurate) information about the current \emph{failure pattern}, e.g., 
a list of processes predicted to take only finitely many steps in the
current execution. 
The failure detector abstraction gives a language for capturing the
weakest support from the system one may require in order
to solve a given task. 
This gave many interesting insights on the nature of ``wait-free unsolvable'' tasks, 
starting from the celebrated result by Chandra et al. on the weakest
failure detector for consensus~\cite{CHT96}.\footnote{Informally, $\D$ is the
  weakest  failure detector to solve a task $T$ if it (1) solves $T$ and (2) can
  be deduced from  any failure detector that solves $T$.}     

A solution of the task using a failure detector  guarantees that  every
\emph{correct} (a process that is predicted to take infinitely many
steps by the failure pattern) eventually obtains an output. 
The progress of each process may thus depend on the behavior of other
correct processes, and therefore failure detector-based algorithm cannot be wait-free. 
Consequently, since the failure pattern is introduced as a part of a
run, we cannot treat individual progress as a liveness property
anymore: a process is not allowed to take steps after it crashes.




\paragraph{Wait-freedom with advice.}
But can we think of a system where a ``hard'' task can be solved so that progress of a
process does not depend on the execution speeds of other processes?
A straightforward way to achieve this is to assume that the processes receive
\emph{advice} from an external oracle, and an immediate question is
what is the weakest oracle that allows for solving a given task so that      
every participating process taking enough steps outputs.  

\remove{
This paper proposes a new look of what does it mean to solve
a task with a failure detector, which enables treating the set of  
participating processes as a wait-free system.
An important consequence of this definition is that it 
allows the processes to take steps on behalf of each other, e.g.,
enables simulations.
}

In this paper, we use the language of failure detectors to determine 
the relative power of such external oracles.
The oracle is represented as a set of \emph{synchronization} processes
equipped with a failure detector: each synchronization process can
query its failure detector module to get hints about the failures of
other synchronization processes. 
Thus, our system only considers failures of synchronization processes.
As in the classical failure-detector literature~\cite{CHT96}, the assumptions about when and where failures of
synchronization processes can  occur are encapsulated in
an \emph{environment}, i.e., a set of allowed failure patterns.
\emph{Computation} processes (participants in a task solution) and synchronization processes 
communicate by reading and writing in the shared memory.


Now what do we mean by solving a task with a failure detector?
We require that, under the condition that the synchronization
processes using their failure detector behave as predicted by the environment, 
every computation  process taking enough steps must output. 

%
%

It is easy to see that the classical failure-detector
model~\cite{CHT96} is a special case of our model where there is a
bijective map between computation and synchronization
processes, and  a computation process stops taking steps after 
its synchronization counterpart does.      
Strictly speaking, when it comes to solving tasks, our framework demands from 
a failure detector more than the conventional failure detector model
does.
Indeed, in our framework, the failure of a synchronization process does not affect
computation processes, and  
a failure detector is supposed to help computation processes
output, as long as they take enough steps.
In particular, we observe that the \emph{weakest} failure detector to
solve a task $T$ in our framework is at least as strong as
the weakest failure detector for $T$ in the conventional model~\cite{CHT96}.


\paragraph{Ramifications.}
The idea of separating computation from synchronization is not new,
e.g., it is used in the celebrated Paxos protocol~\cite{Lam98} 
separating \emph{proposers} from \emph{acceptors} and \emph{learners}.
But applying it to distributed computing with failure detectors 
results in a surprisingly simple model, which we call \emph{external
  failure detection (EFD)}, which resolves a number 
of long-standing puzzles.

\remove{
To warm up, consider the task of 
solving consensus among every two (computation) processes
in a system of $n>2$.
Delporte et al.~\cite{DFG10} showed that any failure detector
that allows for solving the task, 
also allows for solving consensus among \emph{all} $n$ processes. 
Does the phenomenon highlighted in~\cite{DFG10} only hold for
$1$-set agreement (consensus) or  can it be generalized to any $k\geq 1$?
Suppose that a failure detector  provides enough synchrony to solve $k$-set agreement
among any $k+1$ processes. 
A natural conjecture would be that it also allows for solving $k$-set 
agreement among all.
However, years of trying to prove the conjecture
bore no fruits. 

In EFD,  we obtain the generalization of~\cite{DFG10} to any $k\geq 1$ 
almost for free, using simple induction. 
In fact, even a stronger result holds:
if a failure detector solves $k$-set agreement among an arbitrary given set of $k+1$ processes, 
then it is strong enough to solve $k$-set agreement among all $n$ processes.
}

The use of EFD enables a complete characterization of
distributed tasks, based on the ``amount of concurrency'' they can
stand.     
\remove{
First we observe that the \emph{weakest} failure detector to solve a task in EFD is at least as strong as
the weakest failure detector in~\cite{CHT96}.
Indeed, any failure detector that solves the task in EFD, also solves
it in any its restriction, including the conventional failure detector model~\cite{CHT96}.
But since EFD allows for a simulation of computation processes, 
	the power of failure detectors to solve tasks can be completely characterized.
}
%
In the classical framework, we 
say that a task $T$ can be solved $k$-concurrently if it guarantees 
that in every $k$-concurrent run every process taking
sufficiently many steps eventually outputs~\cite{GG11}.     
Informally, a run is $k$-concurrent if at each moment of time,
there are at most $k$ participating processes without outputs.
Now, in a system of $n$ processes, each task $T$ is associated with
the largest $k$ ($1\leq k\leq n$) such that $T$ can be solved $k$-concurrently.

We show that in EFD, a failure detector $\D$ can be used to solve a task $T$ with
``concurrency level'' at most $k$ \emph{if and only if} $\D$ 
	can be used to solve $k$-set agreement.
More precisely, we show that, in every environment, 
	i.e., for all assumptions on where and when 
	failures of \emph{synchronization} processes may occur, 
	any failure detector that solves $T$ is at least as strong as  
  	the anti-$\Omega$-$k$ failure detector~\cite{Ray07,Z10}, denoted $\Nomega_k$.
Then we describe an algorithm that
	uses $\Nomega_k$ to solve $T$ (or any task that can be solved  $k$-concurrently),
	in every environment.

Thus,  any task is completely characterized through the
``level of concurrency'' its solution can tolerate. 
All tasks that can be solved $k$-concurrently
but not $(k+1)$-concurrently (e.g., $k$-set agreement) are equivalent 
in the sense that they require exactly the same amount of information
about failures (captured by $\Nomega_k$) to be solved in EFD.
Note that this characterization covers \emph{all} tasks, including ``colored'' ones 
that evaded any characterization so far~\cite{DFGT11,GK11-setcon,AN08}.

Consider, for example, the task of $(j,\ell)$-renaming in which 
$j$ processes come from a large set 
of potential participants and choose new names in a smaller name space $1,\ldots,\ell$, 
so that no two processes choose the same name.
Surprisingly, in the conventional model, the renaming task itself 
can be formulated as a failure detector, so the question of the weakest failure detector 
for solving it results in a triviality.
To avoid trivialities, additional assumptions on the scope of failure detectors are made~\cite{AN08}. 

In EFD, however, we immediately see that 
$(j,j)$-renaming (also called \emph{strong} renaming) cannot be solved
$2$-concurrently and   is thus equivalent to consensus.\footnote{Note that all tasks can be solved $1$-concurrently.}
More generally, determining the weakest failure detector for $(j,\ell)$-renaming 
boils down to determining the maximal $k$ ($1\leq k\leq j$) 
such that the task can be solved $k$-concurrently.
We show finally that $(j,j+k-1)$-renaming can be solved $k$-concurrently, and, thus,   
using $\Nomega_k$.\footnote{ 
For some values of $j$ and $k$, however, the question 
of the maximal tolerated concurrency of $(j,j+k-1)$-renaming 
is still open~\cite{CR10}.}  

Another interesting corollary of our characterization is 
that if a failure detector solves $k$-set agreement among an arbitrary given subset of $k+1$ processes, 
then it is strong enough to solve $k$-set agreement among \emph{all}
processes.
This is a generalization of the recent result of Delporte et al.~\cite{DFG10} that any failure detector
allowing for solving consensus ($1$-set agreement) among each two processes, 
also allows for solving consensus among all processes. 
Years of trying to show that the phenomenon demonstrated
in~\cite{DFG10} generalizes to all $k\geq 1$ in the conventional
failure-detector model~\cite{CHT96} bore no fruits.

One important feature inherited by our EFD framework from  wait-free protocols is that 
it leverages \emph{simulation-based} computing: processes can
cooperate trying bring \emph{all} participating processes to their outputs.
Simulations were instrumental in establishing tight relations between 
seemingly different phenomena in \emph{asynchronous} 
systems~\cite{BG93b,BGLR01,Gafni09,GK10,DFGT11,GK10-adv,GG11},
and we extend this line of research below to failure-detector models.

\paragraph{Roadmap.}
The paper is organized as follows. 
First, we formally define our model and our new notion of task solvability with a
failure detector. 
We then present a simple inductive proof of a generalization
of~\cite{DFG10} to any $k>1$.
Then we extend the generalization even further by presenting 
a complete characterization of decision tasks, based on the level of
concurrency they can tolerate.
Then we derive the weakest failure detector for strong renaming and  
wrap up with obligatory concluding remarks.
Proofs are partially delegated to the optional Appendix.



\section{The model of external failure detection}
\label{sec:model}
%
In this section, we propose a new definition of what it means
to solve a task using a failure detector
and relate it to the conventional definition of~\cite{CHT96}. 
Parts of our model reuse elements of~\cite{CHT96,CT96,GK11-setcon,HS99}.

\subsection{Model for computation and synchronization}

Our system is split in two parts.
The \emph{computation} part is made up
of processes that get input values for the task they intend to
solve and return output values. 
The \emph{synchronization} part is made up of processes that
use failure detectors to help processes of the computation part. 

%
\vspace{2mm}\noindent\textit{Processes.}
Formally, we consider a read-write shared-memory system which consists of $\nbwf$ \emph{$\wf$-processes},
$\Pi^{\wf}=\{p_1,\ldots,$ $p_\nbwf\}$, and $n$ \emph{$\fd$-processes},
$\Pi^{\fd}=\{q_1,\ldots,q_n\}$. 
We allow $n$ and $m$ to be arbitrary natural numbers, 
but, as we shall see shortly, the only ``interesting'' case
is when $n=m$.

Intuitively, the $\wf$-processes are responsible for computation. 
The  $\fd$-processes are responsible for synchronization and may be equipped 
with a  failure detector module~\cite{CT96} that gives hints about failures of other $\fd$-processes. 
The processes in $\Pi^{\wf}]\cup\Pi^{\fd}$ communicate with each other
via reading and writing in the shared memory. 

\vspace{2mm}\noindent\textit{Failure patterns and failure detectors.}
Since $\wf$-processes are assumed to be wait-free, 
we are only interested here in failures of $\fd$-processes. 
Hence a  \textit{failure pattern} $F$ is a function from the time range $\Time=\Nat$ to 
	$2^{\Pi^{\fd}}$, where $F(\tau)$ denotes the set of $\fd$-processes 
	that have crashed by time $\tau$. 
Once a process crashes, it does not recover, i.e., $\forall \tau: F(\tau) \subseteq F(\tau+1)$.
$\faulty(F)=\cup_{\tau \in \Time} F(\tau)$ is the set of faulty processes in $F$ and
$\correct(F)=\Pi^\fd -\textit{faulty}(F)$ is the set of correct processes in $F$.

A \emph{failure detector history $H$ with range $\R$} is a function from
	$\Pi^{\fd}\times \Time$ to $\R$. 
$H(q_i,\tau)$ is interpreted as the value output by the failure detector module of $\fd$-process $q_i$ at time $\tau$.
A \textit{failure detector} $\D$ with range $\R_{\D}$ is a function that maps each failure 
pattern to a (non-empty) set of failure detector histories with range $\R_{\D}$. 
$\D(F)$ denotes the set of possible failure detector histories permitted by 
	\D~for failure pattern $F$.

An \emph{environment} $\E$ is a set of failure patterns that describes a set of conditions on when 
and where failures might occur.
For example $\E_t$ is the environment that consists 
of all failure patterns $F$ such that $\correct(F)\ge n-t$.
We assume that for every failure pattern in the environments we consider, at least one $\fd$-process is correct.

\vspace{2mm}\noindent\textit{Algorithms and runs.}
A distributed algorithm $\A$ using a failure detector $\D$ 
consists of two collections of
deterministic automata, $\A_1^{\wf},\ldots,\A_\nbwf^{\wf}$, one automaton for each $\wf$-process, and
$\A_1^{\fd},\ldots,\A_n^{\fd}$, one automaton for each $\fd$-process.
In a \emph{step} of the algorithm, a process may read or write to a
shared register, or (if it is a $\fd$-process) consult its
failure-detector module.    

A \textit{state} of $\A$ is defined as the state of each 
	process (state of each process being identified with the state of its corresponding automaton)
        and each shared object in the system. 
An \emph{initial state} $I$ of $\A$ specifies 
	an initial state for every process
	and every shared object. 

A \textit{run of  $\A$ using a failure detector ${\D}$} in an environment 
	$\E$  is a tuple $R=\langle F,H,I,\textit{Sch},T \rangle$ 
	where $F\in\E$ is a failure pattern, 
	$H\in {\D}(F)$ is a failure detector history, $I$ is an initial state of $\A$, 
	$\textit{Sch}$ is an infinite \emph{schedule}, i.e., a sequence of
        processes in $\Pi^{\wf}\cup\Pi^{\fd}$, $T$ is a sequence of
        non-decreasing elements of $\Time$.
The $k$-th step of run $R$ is a step of process $\textit{Sch}[k]$ determined by
the current state, the failure history $H$, $T[k]$ and the algorithm $\A$.
If it is a step of a 
	$\fd$-process, this process is alive
 	 ($\textit{Sch}[k] \notin F(T[k])$) and   the value of the failure detector for this step is
  	given by $H(\textit{Sch}[k],T[k])$.



Let $\infi^{\fd}(R)$ denote the set of processes 
	in  $\Pi^{\fd}$ that appear infinitely often in $\textit{Sch}$.
Respectively,  $\infi^{\wf}(R)$ denote the set of processes in
  $\Pi^{\wf}$ that appear infinitely often in $\textit{Sch}$.
We say that a run $R=\langle F,H,I,\textit{Sch},T  \rangle$ is \emph{fair} 
	if $\correct(F)$ is equal to  $\infi^{\fd}(R)$,
	and $\infi^{\wf}(R)$ is not empty.
A \emph{finite run} of $\A$ is a ``prefix'' of a run $\langle
F,H,I,\textit{Sch},T \rangle$  of $\A$, i.e., a
tuple $\langle F,H,I,\textit{Sch}',T' \rangle$ such that 
$|\textit{Sch}'|=|T'|$,  $\textit{Sch}$ is a proper prefix of
$\textit{Sch}$, and $T'$ is a proper prefix of
$T$.


\vspace{2mm}\noindent\textit{Tasks.}
%
We focus on a class of problems called \emph{tasks}
that are defined uniquely through inputs and outputs.

A \emph{task}~\cite{HS99} is defined through a set $\I$ of input vectors (one input value for each 
$\wf$-process), 
a set $\O$ of output vectors (one output value for each $\wf$-process),
and a total relation $\Delta:\I\mapsto 2^{\O}$ that associates each input vector 
with a set of possible output vectors.
An input value equal to $\bot$ denotes a \emph{not participating} process and
 $\bot$ output value denotes an \emph{undecided} process.

A  $\nbwf$-vector $L'$ is a \emph{prefix} of a $\nbwf$-vector $L$ if 
	$L'$ contains at least one non-$\bot$ item and for all $i$, 
	$1\leq  i \leq \nbwf$,  either $L'[i]=\bot$ or $L'[i]=L[i]$.
A set $\cal L$ of vectors is 
	\emph{prefix-closed} if 
	for all $L$ in $\cal L$ every prefix of $L$ is in $\cal L$.
	
We assume that each element of $\I$ and $\O$ contains at least one non-$\bot$ item
and also that  the sets $\I$ and $\O$  are prefix-closed. 
Moreover, we only consider tasks that have finite sets of input
vectors $\I$ (this assumption is used in Section~\ref{sec:main} when
we categorize tasks based on the failure detectors needed to solve them).
	
We stipulate that if $(I,O)\in\Delta$, then
	(1) if, for some $i$, $I[i]=\bot $, then $O[i]=\bot$,
	(2) for each $O'$, prefix of $O$,  $(I,O')\in\Delta$ and,
	(3) for each $I'$ such that $I$ is a prefix of $I'$, there exists some $O'$ such that $O$ is a prefix of $O'$ and $(I',O')$ in $\Delta$.
	

For example, in the task of \emph{$(U,k)$-agreement}, where $U\subseteq\Pi^{\wf}$, 
	input and output vectors are $m$-vectors, 
	such that $I[i]=\bot$ for all $p_i\notin U$,
	input values are in $\{\bot, 0,\ldots,k\}$,
	output values are in $\{\bot,0,\ldots,k\}$,
	and for each input vector $I$ and output vector $O$, $(I,O) \in\Delta$ if 
	the set  of non-$\bot$ values in $O$ is 
	a subset of values in $I$ of size at most $k$.
$(\Pi^{\wf},k)$-agreement is the conventional \emph{$k$-set agreement} 
task~\cite{Chaud93} and $(\Pi^{\wf},1)$-agreement is \emph{consensus}~\cite{FLP85}.

%
\remove{
We also restrict ourselves to task solvable by consensus.
That is, we exclude tasks that, say, can only be solved \emph{$t$-resiliently}, 
i.e., assuming that 
at most $t$ $\wf$-processes take finitely many steps.
The output of such task may depend on processes observing at least $n-t$
inputs. Such tasks as mentioned later are handled by the notion of environment.
They are not solvable even by consensus.
In this paper we deal only with the trivial environment, which can 
be specified as tasks solvable by consensus.
}

\subsection{Solving a task in the EFD framework}

Now we are ready to define what does it mean to solve a task 
in the external failure detection framework.

\vspace{2mm}\noindent\textit{Input vector and output vector of a run.} First, 
we assume that each automaton $\A^{\wf}_i$ (1)~gets
an input value $input_i$ as part of its initial state, and 
	(2)~contains $decide$ steps such that all the next steps of
	 $\A_i$  are null steps that do not affect the current state
         when they are executed and for each $decide$ step is
         associated a decision value 
         $v_i$.

The first step of each $\wf$-process is to write 
its input value to shared memory. A process that wrote its
input value is called  
\emph{participating}. 
If a $\wf$-process  executes a $decide$ step with decision value $v$, we say
that the process decides $v$
or returns $v$.

Given a run $R$,
the \emph{input vector} for the run is the  $\nbwf$-vector $I$
such that $I(i)=input_i$ if $p_i$ is a
participating process and  $I(i)=\bot$ if  $p_i$ is a not
participating process.
 In the same way, the \emph{output vector} of the run is the
 $\nbwf$-vector $O$ such that $O(i)=v$ if $p_i$ decides $v$
 in the run and $O(i)=\bot$ if $p_i$ does not decide in the run.

\vspace{2mm}\noindent\textit{Solving a task.} 
We say that a run $R$ with input vector $I$ and output vector $O$
\emph{satisfies a task $T=(\I,\O,\Delta)$} if (1) $(I,O)\in \Delta$ and (2)
$O(i)=\bot$ only if $p_i$ makes a finite number of steps ($p_i \notin \infi^{\wf}(R)$).

An
algorithm $\A$ \emph{EFD-solves a task $T=(\I,\O,\Delta)$ using a failure
	detector $\D$ in an environment $\E$} (in the rest we simply say ``solves'') 
	if every fair run of
      $\A$ satisfies $T$. 
If such an algorithm exists for task $T$, $T$ is \emph{solvable
with failure detector $\D$ in
environment $\E$}.
        By extension,  a failure detector $\D$\emph{ solves a task $T$ in $\E$} if there 
	is an algorithm $\A$ that solves $T$ using $\D$ in $\E$.

Note that we expect the algorithm to guarantee output to every 
	$\wf$-process that takes sufficiently many steps,
	regardless of where and when \emph{$\fd$-processes} fail. 
The algorithm only expects that every correct \emph{$\fd$-process} in the current failure pattern 
 	takes infinitely many steps.

%
\vspace{2mm}\noindent\textit{Comparing failure detectors.}
Failure detector reduction is defined as usual:
failure detector $\D'$ is \emph{weaker than failure detector
$\D$ in an environment $\E$} if \fd-processes can use $\D$ 
to emulate $\D'$ in $\E$.
More precisely, 
the  automata of the $\wf$-processes of the distributed \emph{reduction algorithm}  $\A$ 	
	are automata with only null steps and the
	emulation of  $\D'$ using $\D$ is made by maintaining, at each $\fd$-process $q_i$ 
	$\D'\id{-output}_i$ so that in any fair run with failure pattern $F$, 
	the evolution of variables $\{\D'\id{-output}_i\}_{q_i\in\Pi^{\fd}}$ 
	results in a history $H'\in\D'(F)$.
We say that two failure detectors are \emph{equivalent} in $\E$ if each is weaker than 
	the other in $\E$.


As in the original definiton~\cite{CHT96},
      if failure detector $\D'$ is weaker than failure detector $\D$ in
      environment $\E$, then every task solvable with $\D'$ in $\E$
	can also be solved with $\D$ in $\E$.
Now $\D$ is the \emph{weakest failure detector} to solve a task $T$ 	
	in $\E$ if (i) $\D$ solves $T$ in $\E$  and (ii) $\D$ is weaker than 
	any failure detector that solves $T$ in $\E$.
It is straightforward to extend the arguments of~\cite{JT08}
	to show that every task has a weakest failure detector.

\vspace{2mm}\noindent\textit{$k$-concurrency.}
Consider the solvability of a task without the help of a failure detector.
In this case  the deterministic automata of the $\fd$-processes of the distributed algorithm $\A$ 	
	are automata with only null steps.
Such an algorithm will be called \emph{restricted}. 

It is clear that tasks that are
solvable with a restricted algorithm are exactly tasks that are
said \emph{wait-free} solvable in the literature
(e.g.  in \cite{Her91,HS99}).

The notion of \emph{$k$-concurrent} \label{sec:kconc}
	solvability, introduced in~\cite{GG11}, is a weaker form 
	of solvability:
	a  task is solvable $k$-concurrently if it is solvable only 
        when at most $k$ $\wf$-processes  concurrently 
	invoke the task.  
More precisely,
a run of a distributed  algorithm is  \emph{$k$-concurrent} if it 
is fair and at each time there is at most $k$
undecided participating $\wf$-processes.
A task  $T=(\I,\O,\Delta)$ is  \emph{$k$-concurrently} solvable if there is
	 a restricted algorithm $\A$ such 
	 that all $k$-concurrent runs $R$ of $\A$ 
         satisfy  $T$. 
Note that runs of $\A$ in which the number of participating but not decided $\wf$-processes  
	exceeds $k$ at some point may not satisfy $T$.

A wait-free solvable task is $\nbwf$-concurrently solvable.
%
%
Also, it is easy to show that:
	

\begin{proposition}\label{prop:1conc}
Every task is $1$-concurrently solvable.
\end{proposition}

\vspace{2mm}\noindent\textit{Restriction on the number of $\wf$-processes.}
Trivially, if a task $T$ is solvable with a restricted algorithm then
$T$ is also
solvable with any number of \fd-processes and any failure
detector. 
Reciprocally, consider an algorithm $\A$ solving a task $T$ 
 with a trivial failure
detector\footnote{A trivial failure detector 
always outputs $\bot$.}
 in environment
$\E_{n-1}$.
If $ n\geq \nbwf$ consider the following algorithm:
 each \wf-process $p_i$ executes alternatively steps of
 $\A_{p_i}^\wf$ and steps of $\A_{q_i}^\fd$ and each $\fd$-process 
 executes only null steps. 
 It is easy to verify that in this way we
 	emulate runs of $\A$ in 
	the failure pattern in which at least all \fd-processes $q_i$ with
 $i>m$ are crashed, and such runs satisfy task $T$. Hence we get:

\begin{proposition}
If $n\geq m$,
$T$ is solvable  in $\E_{n-1}$ with a trivial
failure detector if and only if $T$ is solvable with a restricted algorithm.
\end{proposition}

But if $n<\nbwf$, the $\fd$-processes may help solving the task even if
        they do not use 
	their failure detection capacities.
For example, with $n$ $\fd$-processes we can implement a
	$(\Pi^{\wf},n)$-set agreement in every environment. 
For this,
each $\fd$-process waits until at least one $\wf$-process writes its
input in shared memory, and then it writes this value to a shared variable
$V$.
Each $\wf$-process waits until $V$ has been written and outputs the
read value. As at least one \fd-process is correct, eventually $V$
will be written and  as there are $n$ $\fd$-processes at most $n$ values may be
output.
In this way the $(\Pi^{\wf},n)$-set agreement is always solvable even
without the help of any failure detector.

As we focus here on solvability where additional power of processes
is only due to the failure detection, 
the  only ``interesting" scenario to consider is when the
number of \wf-processes does not exceed the number of \fd-processes
and more specifically  the case where they are equal.
\emph{Therefore, in the following we assume that the number of \wf-processes is
equal to the number of \fd-processes, we denote this number by $\nbfinal$.}








\subsection{Conventional solvability}\label{sec:solvetask}


More conventional models of computation in which there is no separation
between the computation and the synchronization part may be considered
as a special case of the generalized model presented here.
In conventional models, each process $i\in\{1,\ldots,n\}$ can be seen as running 
two parallel threads:  $p_i$ corresponding to the computational 
part and $q_i$ corresponding to the synchronization part.
Moreover failure patterns correspond: $i$ is correct in
conventional systems if and only if
$q_i$ is correct in our setting.
But, since in our model, computation and synchronization are separate, 
it is possible that $p_i$ makes only a finite number of steps even if $q_i$
is correct or vice-versa.
Then we define \emph{personified runs} of a distributed algorithm as
being runs $R$ that are  fair and such that 
$p_i$ crashes if and only if $q_i$ crashes at the same time
(as a result, $\infi^{\wf}(R)$ is equal to  $\infi^{\fd}(R)$).
We say that algorithm $\A$ solves \emph{classically} task $T$ with failure detector $\D$ in 
	environnent $\E$
	if  every personified run $R$ of $\A$ satisfied $T$.

This definition corresponds exactly to the notion of solvability in a
conventional setting as can be found in the literature~\cite{CHT96}.

As the set of personified runs of a distributed algorithm is a  subset of the fair runs,  we have:

\begin{proposition}\label{prop:solveTOclassic}
If a failure detector $\D$ solves a task $T$ in an environment $\E$ then $\D$ classically solves 
$T$ in  $\E$.
\end{proposition}

\begin{corollary}\label{col:WFDsolveTOclassic}
If $\D$ is the weakest failure detector to classically solve a task $T$ in an environment $\E$, then 
	$\D$ is weaker than the weakest  failure detector to solve $T$ in $\E$. 
\end{corollary}	
%
%
Note that the converse of Proposition~\ref{prop:solveTOclassic} is not
true. For example, consider  
the $(\{p_1,p_2\},1 )$-agreement task (consensus among $p_1$ and $p_2$).
It is classically solvable in $\E_2$ (assuming at most $2$ failures) 
with the failure detector $\D$  that, for each \fd-process, outputs
 $q_1$ if $q_1$ is correct and outputs $q_2$ if $q_1$ is faulty.
But this task is not solvable in $\E_2$ with this failure detector
(intuitively, otherwise, if $q_1$ is crashed we would be able to solve consensus between
$p_1$ and $p_2$ without a failure detector).

However for \emph{colorless} tasks\footnote{Informally, in a solution of 
a colorless task~\cite{BGLR01}, 
a process is free to adopt the input
or
the output value of any other participating process.}  both notions of solvability coincide.

\begin{proposition}\label{prop:colorless}
Let $T$ be a colorless task, $T$ is  solvable with   failure detector $\D$ in environment $\E$ if and only if  $T$ is classically solvable
	with  $\D$ in  $\E$.
The weakest failure detector to solve $T$ in $\E$ is the weakest failure detector to classically solve $T$ in $\E$.	
\end{proposition}
	

\vspace{2mm}\noindent\textit{Failure detectors for $k$-set agreement.}

The failure detector $\Nomega_k$~\cite{Z10} outputs, at every \fd-process and each time, 
a set of $(n-k)$ \fd-processes. 
$\Nomega_k$ guarantees that there is a time after which  
some correct \fd-process is never output:
\[
\begin{array}{l}
\forall F \in \E,\; \forall H\in \Nomega_k(F),\; \exists q_i\in\correct(F),\; \tau\in\Time,\\ 
\hspace{4mm}	\forall \tau'>\tau, \forall q_j\in\correct(F): q_i\notin H(q_j,\tau'). 
\end{array}
\]

$\Nomega_{1}$ is equivalent to $\Omega$~\cite{CHT96} that outputs a $\fd$-process  
such that eventually the same correct $\fd$-process is permanently output at all correct processes.

From \cite{GK11-setcon}, we know that in every environment $\E$, 
$\Nomega_k$ is the weakest failure detector 
to classically solve $(\Pi^{\wf},k)$-set agreement in $\E$.
As  $(\Pi^{\wf},k)$-set agreement is a colorless task, from Proposition~\ref{prop:colorless} we obtain:

\begin{proposition}\label{prop:antiOmega}
In every environment $\E$, $\Nomega_k$ is the weakest failure detector 
to solve $(\Pi^{\wf},k)$-set agreement in $\E$. 
\end{proposition}
\section{Solving a puzzle}
\label{sec:kset}

Let $U$ be a set of $k+1$ \wf-processes.   
Consider a failure detector $\D$ that \emph{solves} 
$k$-set agreement among the processes in $U$.
We show that $\D$ actually solves $k$-set agreement among \emph{all} $\nbfinal$
\wf-processes. 

\begin{theorem}
\label{th:ind}
Let $U$ be a set of $(k+1)$  \wf-processes, for some $1\le k<\nbfinal$.   
For every environment  $\E$, if a failure detector $\D$ solves $(U,k)$-set agreement 
in~$\E$ then $\D$ solves $(\Pi^{\wf},k)$-set agreement in~$\E$.  
\end{theorem}
\begin{proofsketch}
Without loss of generality, assume that $U=\{p_1,\ldots,p_{k+1}\}$.
%
Let $\A$ be a distributed algorithm that solves the 
$(U,k)$-set agreement  in $\E$ with $\D$.

Let $U_x$ denote $\{p_1,\ldots,p_x\}$, $x=k+1,\ldots,\nbfinal$.
We observe first that $\D$ can be used to solve $(U_x,x-1)$-set
agreement as follows.
$\wf$-processes in $\{p_1,\ldots,p_{k+1}\}$ and $\fd$-processes  $\{q_1,\ldots,q_n\}$ run $\A$ to solve $k$-set agreement and 
return the value  returned by the algorithm,
and processes in $\{p_{k+2},\ldots,p_x\}$ simply return their own input values.
In total, at most $x-1$ distinct input values are returned.
Let $\A_x$ denote the resulting algorithm. 

We proceed now by downward induction to show that 
for all $x=\nbfinal$ down to $k$, $\D$ solves $(\Pi^{\wf},x)$-set agreement. 

The base case is immediate:  $\{p_1,\ldots,p_{\nbfinal}\}$  trivially
solve $(\Pi^{\wf},\nbfinal)$-set agreement without any 
failure detector. 
Now suppose that $\D$ solves $(\Pi^{\wf},x)$-set agreement for $x\ge k+1$.
By Proposition~\ref{prop:antiOmega}, $\D$ can be used 
to implement $\Nomega_{x}$.
 
Using the generic simulation technique presented in
Appendix~\ref{app:kcodes}, 
the \wf-processes, $p_1,\ldots,p_{\nbfinal}$, can use $\Nomega_{x}$ 
to simulate a run of the \wf-part of $\A_{x}$ on $p_1,\ldots,p_{x}$,
so that at least one simulated process takes infinitely many steps.\footnote{
We could have used a ``black-box'' simulation of  $\A_x$
using $(\Pi^{\wf},x)$-set agreement objects presented in~\cite{GG11}. 
To make the paper self-contained, we give a direct construction 
using $\Nomega_x$ in Appendix~\ref{app:kcodes}.}
The $\fd$-part of $\A_{x}$ is executed by $\fd$-processes.
In the simulation, each simulating process proposes
its input value as an input value in the first step for each 
simulated process in $\{p_1,\ldots,p_{x}\}$ (this can be done, since
$(\Pi^C,x)$-set agreement is a colorless task).

Suppose that the current run is fair, i.e., every correct 
\fd-process takes infinitely many steps.
Therefore, we simulate a fair run of $\A_x$ and  
thus eventually some simulated \wf-process in $\{p_1,\ldots,p_{x}\}$  
decides on one of the input values of the \wf-processes. 
Once a simulator finds out that a simulated process decided, it returns the decided value.
Thus, eventually, every correct simulator returns.
Since all decided values come from a run of $\A_x$, 
at most $x-1$ distinct input values can be decided. 
Hence, $\D$ solves  $(\Pi^{\wf},x-1)$-set agreement. 
\end{proofsketch}
Therefore, in our framework, we obtain a direct generalization of the fact that 
for a failure detector, it is as hard to solve consensus in a system of $\nbfinal$ processes
as to solve consensus among each pair of processes~\cite{DFG10}.
In fact, the separation between \wf-processes and \fd-processes,
implies a stronger result: 
solving $k$-set agreement among one given set of $(k+1)$ processes is as hard 
(in the failure detector sense) as solving it among all $\nbfinal$ processes.

\section{Generalizing the puzzle}
\label{sec:main}

We showed in the previous section that solving $k$-set agreement  among any given set 
of $k+1$ $\wf$-processes requires an amount of information 
about failures that is sufficient to solve $k$-set agreement among all $\nbfinal$ \wf-processes.
We show below that this statement can be extended to any task $T$ that cannot be solved 
$(k+1)$-concurrently.
%
We present an explicit reduction algorithm 
	that extracts $\Nomega_k$ from any failure detector that solves $T$.
Conversely,  we show that a task that is  $k$-concurrently solvable 
	can be solved with $\Nomega_k$ in any environment.


Finally, we derive a complete characterization of generic tasks : 
all tasks that can be solved $k$-concurrently but not 
$(k+1)$-concurrently are equivalent in the sense that they
require the same information about failures to be solved ($\Nomega_k$).   
  
\subsection{Reduction to $\Nomega_k$}
\label{subsec:reduc}
Let $T$ be any task that cannot be solved  $(k+1)$-concurrently. 
Let $\E$ by any environment.
We show that every failure detector $\D$ that solves $T$ in $\E$ can be
used to implement $\Nomega_k$ in $\E$ as follows. 

Let $\A$ be the algorithm that solves $T$ using $\D$ in $\E$.
Recall that $\A$ consists of two parts: $\A^{\wf}$ is run by the
	\wf-processes $p_1,\ldots,p_\nbfinal$ and  $\A^{\fd}$  is run by the
	\fd-processes $q_1,\ldots,q_n$.

First, we construct a restricted  algorithm $\A_{\id{sim}}$.
In $\A_{\id{sim}}$, \wf-processes $p_1,\ldots,p_\nbfinal$ 
perform two parallel tasks. 
In the first task, $\wf$-processes take steps on behalf of $\A^\wf$.
In the second task, they simulate a run of $\A^\fd$ on \fd-processes 
 using, instead of $\D$, 
	a \emph{directed acyclic graph} (DAG) $G$.
The DAG $G$ contains a sample of values output by $\D$ 
	in some run $R$ of $\A$~\cite{CHT96, Z10}.  
In $\A_{\id{sim}}$, \fd-processes take null steps.	

Informally, each run of $\A_{\id{sim}}$ gives ``turns'' to the $\fd$-processes 
	and if $G$ provides enough information about failures to simulate the next step 
	of a $\fd$-process $q_j$, the step of $q_j$ appears in the simulated run of $\A$.    
To simulate steps of $\A^{\fd}$, 
\wf-processes employ BG-simulation~\cite{BG93b,BGLR01}.
This simulation technique  enables  
       $k+1$ processes  called \emph{simulators}, to
	 simulate a run  of any asynchronous $n$-processes 
	protocol in which at least $(n-k)$ processes take infinitely many steps.	
%
%
Thus, if $k$ or less participating $\wf$-processes take a finite number of steps, 
	the resulting run of 
	$\A_{\id{sim}}$ gives infinitely many turns to at least $n-k$ $\fd$-processes.

Let $F$ be the failure pattern of the run in which $G$ was constructed.
$\A_{\id{sim}}$ guarantees that  
(1) every finite run of $\A_{\id{sim}}$ simulates a finite run of $\A$, and 
(2) if every \fd-process that is correct in $F$ 
receives infinitely many turns to take steps, then 
the simulated run of $\A$ is fair, and
(3) if $k$ or less participating $\wf$-processes take only finitely many 
	number of steps, then there are at most $k$ \fd-processes that  
	receive only finitely many turns to take steps in the
        simulation.

Second we construct a reduction algorithm. In such an algorithm \wf-processes take null steps.
Our reduction algorithm consists of two components (both are run exclusively 
by the \fd-processes). 
In the first component, every \fd-process $q_i$ queries $\D$, exchanges the 
	returned values with other 
	\fd-processes and maintains a DAG $G_i$.
In the second component, each $q_i$ locally simulates multiple $(k+1)$-concurrent runs of 
	$\A_{\id{sim}}$ 
	using $G_i$, going over all combinations of inputs, exploring the runs in 
	the depth-first manner. 
The simulation continues as long as some simulated \wf-process does not decide
	in the produced run of $\A_{\id{sim}}$.
Since $T$ cannot be solved $(k+1)$-concurrently, 
	there must be a $(k+1)$-concurrent run of $\A_{\id{sim}}$ 
	in which some participating  \wf-process that takes infinitely many steps never decides. 
The only reason for a \wf-process not to decide in a run of $\A_{\id{sim}}$ 
	is that some correct  \fd-process receives only finitely many turns in
	 the 	simulation.
But in the simulation, at least $(n-k)$ \fd-processes receive infinitely many turns.
Thus, by outputting the identities of the $(n-k)$ \fd-processes that were last to receive turns 
	in the current run we emulate the output of $\Nomega_k$: we 
	output sets of $n-k$ \fd-processes 
	that eventually never contain some 
	correct process.


\begin{theorem}
\label{th:nec}
Let $T$ be a task that cannot be solved $(k+1)$-concurrently. 
For every environment $\E$,  for every failure detector $\D$ that solves
$T$ in~$\E$,  $\Nomega_k$ is weaker than $\D$ in~$\E$.
\end{theorem}

\subsection{Solving a $k$-concurrent task with $\Nomega_k$}\label{subsec:ome}
In this section, instead of  $\Nomega_k$, we use an equivalent failure detector $\Vomega_k$~\cite{Z10}.
Basically, $\Vomega_k$ gives a $k$-vector of processes such that, 
eventually, at least one position of the vector stabilizes on the same correct process
at all correct processes.

By definition if $T$ is $k$-concurrently solvable, then there exists a restricted 
algorithm $\A$ that $k$-concurrently solves~$T$.

First, we define an abstract simulation technique that, with
help of $\Vomega_k$, allows us to simulate, in a system of $\nbfinal$ \wf-processes, runs of any
restricted input-less algorithm on  $k$ \wf-processes (the set of non-$\bot$ input values is a singleton).
Moreover, in this simulation, if $\ell$ simulators participate then at most $min(k,\ell)$ processes take
infinitely many steps in the simulated execution.
Basically, to perform a step for a simulated $\wf$-process $p_i$,
 the \wf-processes and the \fd-processes execute
an instance of a leader-based consensus algorithm~\cite{CT96}, 
	using the item $i$ of  $\Vomega_k$ as a leader.
The property of $\Vomega_k$  ensures that for some $i$, infinitely 
	many consensus instances terminate.

 
Second, we define a restricted algorithm $\B$ 
for $k$ \wf-processes that simulates
a $k$-concurrent run of $\A$, using the BG-simulation techniques~\cite{BG93a,BGLR01} .
Applying the abstract simulation technique to $\B$, 
we obtain an algorithm in which every run $R$ 
simulates 
a run $R_{sim}$ of $\A$ such that:
(1) $R_{sim}$ contains only steps of participating processes of $R$, 
(2) the inputs of the participating processes are the same in $R$ and $R_{sim}$,
(2) $R_{sim}$ is $k$-concurrent, and 
(3) every  \wf-process  that takes infinitely many steps in $R$ takes also infinitely many steps in  $R_{sim}$.
So if $T$ is $k$-concurrent solvable with $\A$, $R_{sim}$ satisfies $T$, 
and, consequently, $R$ satisfies $T$.


To sum up,  we have constructed an algorithm that solves $T$ with  
$\Nomega_k$: with the help of $\fd$-processes and  $\Nomega_k$,
$p_1,..,p_\nbfinal$  simulate $\wf$-processes  $p'_1,..,p'_k$  
that, in turn, simulate 
 $\wf$-processes $p''_1,..,p''_\nbfinal$ taking steps in a
 $k$-concurrent execution of algorithm $\A$.

\begin{theorem}
\label{th:omega}
Let $T$ be any  $k$-concurrently solvable task.
For every environment~$\E$,  $\Nomega_k$ solves $T$ in~$\E$.
\end{theorem}

\subsection{Task hierarchy}
From Theorems \ref{th:nec} and \ref{th:omega}, we deduce:

\begin{theorem}
\label{th:hierarchy}
Let $T$ be a task that can be solved $k$-concurrently
but not $(k+1)$-concurrently. 
In every environment $\E$, $\Nomega_k$ is the weakest failure
detector  to solve $T$ in $\E$.
\end{theorem}
As a corollary, all tasks that can be solved $k$-concurrently
but not $(k+1)$-concurrently (e.g., $k$-set agreement) are equivalent 
in the sense that they require exactly the same amount of information
about failures (captured by $\Nomega_k)$.   


\section{Characterizing the task of strong renaming}
\label{sec:renaming}

To illustrate the utility of our  framework, 
we consider the task of \emph{$(j,\ell)$-renaming}~\cite{ABDPR90}.
The task is defined on $\nbfinal$ ($\nbfinal>j$) processes and assumes that in every run 
at most $j$ processes participate (at least $\nbfinal-j$ elements of 
each vector $I\in\I$ are $\bot$). 
As an output, every participant obtains a unique \emph{name} in the range $\{1,\ldots,\ell\}$ 
(every non-$\bot$ element in each $O\in\O$ is a distinct value in $\{1,\ldots,\ell\}$).

In this section, we first focus on $(j,j)$-renaming 
(also called \emph{strong} $j$-renaming). 
%
Using Theorem~\ref{th:hierarchy},  we show that  the weakest failure detector 
for strong $j$-renaming is $\Omega$  (for each $1 < j <\nbfinal$). 
In other words, strong renaming is equivalent to consensus. 

Note that in strong $2$-renaming  at most $2$ $\wf$-processes concurrently 
execute steps of the algorithm. So the impossibility to achieve strong $2$-renaming 
is equivalent to the impossibility of solving strong $2$-renaming  $2$-concurrently. 
By a simple reduction to the impossibility of wait-free $2$-processes consensus, we 
show (Appendix~\ref{app:renaming}):
\begin{lemma}\label{lem:sr-imp}
Strong $2$-renaming  cannot be solved $2$-concurrently.
\end{lemma}
By reducing to the impossibility of Lemma~\ref{lem:sr-imp}, 
we get a more general result:
\begin{theorem}
\label{th:sr-imp}
For all $1<j<\nbfinal$,  strong $j$-renaming  cannot be solved $2$-concurrently.
\end{theorem}   
Proposition~\ref{prop:1conc}, Theorem~\ref{th:hierarchy}, and Theorem~\ref{th:sr-imp} imply:

\begin{corollary}
For all $j$ $(1< j<\nbfinal$),    in every environment $\E$, 
$\Omega$ is the weakest failure detector for solving strong $j$-renaming in $\E$.
\end{corollary} 
In fact, there exists a generic algorithm (Appendix~\ref{app:solvingrenaming}) 
that, for all $k=1,\ldots,j$, solves $(j,j+k-1)$-renaming in all 
$k$-concurrent runs,
and thus $(j,j+k-1)$-renaming can be solved using $\Nomega_k$.
For some values of $k$ and $j$, $(j,j+k-1)$-renaming can be shown to 
be impossible to solve $(k+1)$-concurrently, for others determining 
the maximal level of concurrency of $(j,j+k-1)$-renaming is still an 
open question~\cite{CR10}.

%
%
%
%
%


\section{Conclusion}
\label{sec:conclusion}


This paper introduces a new model of distributed computing 
with failure detectors that allows processes to cooperate. 
A process in this model is able to advance the computation of
other participating processes in the way used previously only in  
asynchronous simulations~\cite{BG93b,BGLR01,Gafni09,GG11}, while using
failure detectors to overcome asynchronous impossibilities.   
In our new framework, we derive a complete characterization of
distributed tasks, based on their maximal ``concurrency level'': 
class $k$ ($1,\ldots,\nbfinal$) consists of tasks that can be solved 
at most $k$-concurrently, and all tasks in the class are equivalent 
to $k$-set agreement.

Our framework does not have to be tied to wait-freedom. We can 
think of its generalization to any progress condition on computation
processes encapsulated, e.g., in an \emph{adversary}~\cite{DFGT11}. Therefore, we can pose
questions of the kind: what is the weakest failure detector to solve
a task $T$ in the presence of an adversary $\A$?
This gives another dimension to the questions explored in this paper.

\remove{

Our formalism is perfectly suitable for tasks, since 
the correctness of outputs in a task solution is determined solely by the participating set.  
An interesting open question is how to extend it to more general classes 
of distributed computing problems in which
failures of participating processes may affect correctness,  
such as NBAC~\cite{Gra78,DFGHKT04}, FTME~\cite{DFGK05}, etc.

\section*{Acknowledgements}
The comments of anonymous reviewers on an earlier version of this paper 
are gratefully acknowledged.
}




\newpage

\bibliography{references}

\newpage
\appendix
\appendix
\section{Proof for $1$-concurrent solvable (Section~\ref{sec:kconc})}

\newenvironment{result-repeat}[2]{\begin{trivlist}
\item[\hspace{\labelsep}{\bf\noindent {#1}~\ref{#2} }]\it}%
{\end{trivlist}}

\begin{result-repeat}{Proposition}{prop:1conc}
Every task is $1$-concurrent solvable.
\end{result-repeat}
\begin{proof}
Each $\wf$-process $p_i$ executes the following code
(1) writes its input, 
(2) reads the other inputs already written getting a vector $I$ such
that $I[i]=\bot$, and
(3) reads all the other outputs already written getting  a vector $O$.
If $O$ is only composed with $\bot$ then $p_i$ is the first process,  it chooses an 
	output according to its input and $\Delta$.
Otherwise let $I'$ obtained from $I$ by replacing the $i$-th item with the
	input value of $p_i$.
By definition of tasks, if $(I,O)\in\Delta$, there exists a vector
	$O'$ obtained from $O$ by replacing the $i$ item by a non $\bot$ value 
	such that $(I',O') \in \Delta$.
Then $p_i$ decides and outputs value $O'[i]$.
Let $R$ be a 1-concurrent run, by an easy induction on the number of participating processes 
we prove that $R$ satisfies $T$.
\end{proof}

\section{Proof for the reduction to $\Nomega_k$ (Section~\ref{subsec:reduc})}

The algorithm sketched in Figure~\ref{fig:red} describes the
steps to be taken by \fd-processes $q_1,\ldots,q_n$  to emulate $\Nomega_k$.  
First we describe the asynchronous algorithm $\A_{\id{sim}}$ used by the \wf-processes to
simulate runs of $\A$, given a sample of the output of $\D$.
Then we describe how the \fd-processes use multiple simulated runs of
$\A_{\id{sim}}$
to emulate the output of $\Nomega_k$.  


\paragraph{Asynchronous simulation of $\A$.}

Following the technique of Chandra et al.~\cite{CHT96}, we
represent a sample of the failure-detector output in the form of a directed
acyclic graph (DAG). 
The DAG is constructed by the \fd-processes by periodically querying $\D$ and
collecting the output values: 
every vertex of the DAG has the form $[q_i,d,k]$ which conveys 
that the $k$-th query of $\D$ performed by process $q_i$ returned value $d$.
An edge between vertexes $[q_i,d,k]$ and $[q_j,d',k']$ conveys that 
the $k$-th query of $\D$ performed by $q_i$ \emph{causally
precedes}~\cite{Lam78a} the $k'$-th query of $\D$ performed by $q_j$.

As in~\cite{Z10,GK11-setcon}, any such DAG $G$ 
can be used to construct a \emph{restricted} algorithm
$\A_{\id{sim}}$.

In $\A_{\id{sim}}$  the \wf-processes 
 $p_1,\ldots,p_n$  simulates runs of $\A$.
The \wf-processes obtain input values for $T$ and 
perform two parallel tasks.
First, the \wf-processes take steps on behalf of $\A^{\wf}$.  
Second, they use BG-simulation~\cite{BG93a,BGLR01} 
to simulate a run of $\A^{\fd}$ on
$q_1,\ldots,q_n$. But to simulate step of $\fd$-process instead of $\D$ they use the information
provided by $G$. More precisely, in the simulation, every \fd-process $q_i$ takes steps as prescribed by
$\A^{\fd}$,
except that  when $q_i$ is about to query $\D$, 
it chooses the next vertex $[q_i,d,k]$  
causally succeeding the latest simulated steps of $\A^{\fd}$ 
of all \fd-processes seen by $q_i$ so far. 
If $G$ was constructed in a run of $\A$ with failure pattern $F$, 
it is guaranteed that  
 (1) every finite run simulated by $\A_{\id{sim}}$ is a run of
$\A$ with failure pattern $F$,
and 
(2) if the run of $\A_{\id{sim}}$ 
contains infinitely many simulated steps of processes in $\correct(F)$ then the
simulated run is a fair run of  $\A^{\fd}$ with failure pattern $F$~\cite{Z10,GK11-setcon}. 

$\A^{\fd}$ does not have inputs. 
Therefore, the simulation tries to promote all $\nbfinal$ \fd-processes (but succeed to take step for a \fd-process $q_i$ if there is enough value for $q_i$ in $G$).

If the simulated run of $\A$ generates an output value for $p_i$,
$p_i$ outputs this value and leaves the computation.   
Note that since $T$ cannot be solved $(k+1)$-concurrently, and all runs of $\A$ are safe,
there must be a $(k+1)$-concurrent (simulated) run of $\A$ in which some participating process
takes infinitely many steps without outputting a value.

\begin{figure}[tbp]
\hrule \vspace{1mm} {\small
\begin{tabbing}
 bbb\=bb\=bb\=bb\=bb\=bb\=bb\=bb\=bb\=bb\=bb \=  \kill
\nnll\> \textbf{for all} $I_0$, input vectors of $T$ 
	(in some order) \textbf{do} \\
\>\>\>\>\>\>\>  /* \texttt{All possible inputs for $p_1,\ldots,p_{n}$} */ \\
\nnll\>\> \textbf{for all} $\pi_0$, permutations of $p_1,\ldots,p_n$  
	(in some order) \textbf{do} \\
\>\>\>\>\>\>\> /* \texttt{All possible ``arrival orders''} */ \\
%
\nnll\>\>\> $P_0 := $ the set of first $k+1$ \wf-processes in $\pi_0$\\
\nnll\label{line:explore0}\>\>\> \textit{explore}$(I_0,\bot,P_0,\pi_0)$\\
\\
\nnll\>\textbf{function} \textit{explore}$(I,\sigma,P,\pi)$ \\
\nnll\label{line:output}\>\> $\Nomega_k\id{-output}_i :=$ $n-k$ \fd-processes that appear the latest 
in $\alpha_i(I,\sigma)$\\
\>\>\>\>\> (any $n-k$ \fd-processes if not possible)\\ 
\nnll\label{line:decide}\>\> \textbf{if} $\exists q_j\in\Pi^{\fd}$: 
	$\forall \sigma'\in \textit{dom}(\alpha_i)$, 
	$\exists \sigma''$, a prefix of $\sigma'$: $\alpha_j(I,\sigma'')$ 
	is deciding \textbf{then}\\
\>\>\>\>\>\>\> /* \texttt{If all schedules explored so far were found deciding by $q_j$} */\\
\nnll\label{line:adopt}\>\>\> adopt $q_j$'s simulation\\  
\nnll\label{line:continue}\>\> \textbf{else}\\
\nnll\>\>\> $N :=$ the set of undecided processes in $(I,\sigma)$\\
\nnll\>\>\> \textbf{for all} $p_j\in P-N$  \textbf{do} 
				\hspace{3mm} /* \texttt{For each decided process in $P$} */\\
\nnll\>\>\>\> $P:=P-\{p_j\}$\\
\nnll\label{line:replace}\>\>\>\> $P:=P\cup \{$the first process in $\pi$ that does not appear in $\sigma\}$\\
\>\>\>\>\>\>\> /* \texttt{Replace $p_j$ with the next non-participant in $\pi$} */\\ 
\nnll\label{line:corridor}\>\>\> \textbf{for all} $P'\subseteq P$  (in
some order consistent with $\subseteq$) \textbf{do} \\ 
\>\>\>\>\>\>\> /*\texttt{For all ``sub-corridors''} */ \\
\nnll\label{line:nextstep}\>\>\>\> \textbf{for all} $p_j\in P'$ (in $\pi$) \textbf{do}\\
\nnll\label{line:explore}\>\>\>\>\> \textit{explore}$(I,\sigma\cdot p_j,P',\pi)$
\end{tabbing}
\vspace{-2mm}\hrule }
\caption{\small Deriving $\Nomega_k$: code for each \fd-process $q_i$.}
\label{fig:red}
\end{figure}

\paragraph{Extracting $\Nomega_k$.}

Now to derive $\Nomega_k$, each $\fd$-process in $i\in\{1,\ldots,k\}$ collects
the output of $\D$ in $G$ and simulates locally multiple
$(k+1)$-concurrent runs of
$\A_{\id{sim}}$.
The runs are simulated in the \emph{corridor-based} depth-first manner~\cite{GK11-setcon}
that works as follows.

We assume a total order on the subsets $P\subseteq \Pi^{\wf}$ so 
that if $P\subset P'$ then $P$ precedes $P'$ in the order. 
Each initial state $I$ and each \emph{schedule} $\sigma$, a sequence 
	specifying the order in which $p_1,\ldots,p_{n}$ take steps of $\A_{\id{sim}}$, 	
	determine a unique run of $\A_{\id{sim}}^{\fd}$ simulated at process $q_i$, denoted 
	$\alpha_i(I,\sigma)$.

   
For a given input vector $I$ and a given permutation $\pi$ of $p_1,\ldots,p_n$, that describes 
the order in which the \wf-processes ``arrive'' at the computation.
Initially, we select a set $P$ of the first $k+1$ processes in $\pi$ as the participating set.
Subsets $P'\subseteq P$ are then explored as ``corridors'' (line~\ref{line:explore}), 
in the deterministic order, from the narrowest (solo) corridors to wider and wider ones.
Recursively, we go through simulating all runs in which only \wf-processes in $P'$ take steps.
In the course of simulation, if a participating  \wf-process $p_j$ decides, we replace it 
with a process that has not yet taken steps in the current computation (line~\ref{line:replace}).
Since we only replace a decided process with a ``fresh'' non-participant, the participating set 
keeps the size of $k+1$ or less processes. 
This procedure is repeated until every \wf-process decides.
Thus, every simulated run is $(k+1)$-concurrent.
Once the exploration of the current corridor is complete (the call of \textit{explore} in 
line~\ref{line:explore} returns), we proceed to the next corridor, etc. 

If, at some point, $q_i$ finds out that another \fd-process $q_j$ made more progress 
in the simulation (simulated more runs than $q_i$),  then $q_i$ ``adopts'' the simulation of $q_j$ 
(line~\ref{line:adopt}) by adopting $q_j$'s version of the DAG and the map $\alpha_j$~\cite{GK11-setcon}.

The output of $\Nomega_k$ is evaluated as the set of the 
ids of the latest $n-k$ processes in $q_1,\ldots,q_n$ that appear in 
the run of $\A_{\id{sim}}^{\fd}$ in the currently simulated run of $\A_{\id{sim}}$ (line~\ref{line:output}).

Recall that $T$ cannot be solved $(k+1)$-concurrently and thus 
there must exist a $(k+1)$-concurrent run of $\A_{\id{sim}}$ in which some participating 
live process never decides.
Since the only reason for the run of $\A_{\id{sim}}$ not to decide is the
absence of some correct process in the simulated $k$-resilient run of $\A_{\id{sim}}^{\fd}$,
and the emulated output 
eventually never contains some correct process---
$\Nomega_k$ is emulated. Thus:

\begin{result-repeat}{Theorem}{th:nec}
Let $T$ be a task that cannot be solved $(k+1)$-concurrently. 
For every environment $\E$,  for every failure detector $\D$ that solves
$T$ in $\E$,  $\Nomega_k$ is weaker than $\D$ in $\E$.
\end{result-repeat}

\begin{proofsketch}
Our reduction algorithm works as follows. 
Every \fd-process $q_i$ runs two parallel tasks.
First, it periodically queries its module of $\D$ and maintains its
directed acyclic graph $G_i$, as in~\cite{CHT96,GK11-setcon}.   
Second, it uses $G_i$ to locally simulate multiple runs of $\A_{\id{sim}}$ and
emulates the output of $\Nomega_k$.
Consider any run of the reduction algorithm.  
Let $F$ be the failure pattern of that run.

First we observe that every simulated run of $\A_{\id{sim}}$ is
$(k+1)$-concurrent. Indeed, initially, exactly $(k+1)$ \wf-processes
participate and a new participant joins only after some participating
\wf-process decides and departs.

Then we show that the correct  \fd-processes eventually perform the same infinite sequence 
	of recursive invocations of 
	\textit{explore}: $\textit{explore}(I,\bot,P_0,\pi)$ invokes $\textit{explore}(I,\sigma_1,P_1,\pi)$, 
	which in turn invokes $\textit{explore}(I,\sigma_2,P_2,\pi)$,
        etc. (line~\ref{line:corridor}).
Indeed, all \fd-processes perform the simulations in the same order and
since, the task is not 
	$(k+1)$-concurrently solvable, there must be a never deciding
        $(k+1)$-concurrent run of $\A_{\id{sim}}$.
Since all these $P_{\ell}$ are non-empty, there exists $\ell^*$ and $P^*$ 
	such that $\forall \ell\geq \ell^*$, $P_{\ell}=P^*$.
Since we proceed from narrower corridors to wider ones, 
$P^*$ is the set of live \wf-processes that never
decide in the  ``first'' never deciding  $(k+1)$-concurrent simulated
run with a schedule $\sigma^*$.

Now we observe that all simulated runs eventually always extend a prefix $\bar\sigma^*$ of $\sigma^*$ in which 
some simulated processes not in $P^*$ already took all their steps in $\sigma^*$. 
Moreover, there is a time after all explored extensions of
$\bar\sigma^*$ only contain steps of processes in $P^*$.
By the properties of BG-simulation~\cite{BG93b,BGLR01}, every
\fd-process that appears only finitely often in the run
of $\A_{\id{sim}}$ simulated by $\sigma^*$ (we called these processes
\emph{blocked} by $\sigma^*$)  eventually never appears 
in all simulated run of  $\A$.
Let $U$ be the set of \fd-processes blocked by $\sigma^*$. 
Since the run of $\A_{\id{sim}}$ simulated by $\sigma^*$ is $(k+1)$-concurrent, 
processes in $U$ eventually never appear  among the last $n-k$ 
processes in $\alpha(I,\sigma)$ (line~\ref{line:output}).

Now we observe that $U$ must contain a correct  (in $F$)  \fd-process.
If it is not the case, i.e., $U$ doesn't contain a correct \fd-process,
then the simulated run of $\A$ is fair and thus the simulated run of $\A$ must be deciding. 

Thus, eventually some correct  \fd-processes never appear in
$\Nomega_k\id{-output}_i$ at every correct \fd-process
$q_i$---$\Nomega_k$ is emulated.
\end{proofsketch}

\section{Proof for solving a $k$-concurrent solvable task with  $\Nomega_k$ (Section~\ref{sec:kconc})}
\label{app:kcodes-kconc}

This section presents a distributed algorithm that uses $\Nomega_k$ to solve, in
any environment, any task that can be solved  $k$-concurrently.
The result could have been obtained from the
simulation of $k$-concurrency using (black-box) $k$-set agreement
objects~\cite{GG11}.
But for the sake of self-containment, we present a (simpler) 
direct construction of a $k$-concurrent run using $\Nomega_k$.\footnote{The construction is similar to the one presented
  in~\cite{GK11-setcon} for 
  the actively $k$-resilient case.}  

First we describe an abstract simulation technique 
that uses $\Vomega_k$ (equivalent to $\Nomega_k$~\cite{Z10}) 
to simulate, in a system of $\nbfinal$ \wf-processes, 
a run of an arbitrary asynchronous algorithm $\B$ on $k$ \wf-processes.

Then we apply this technique to show that, in every environment, 
we  can use $\Vomega_k$ to simulate
a run $R_{sim}$ of any given $\nbfinal$ \wf-processes protocol $\A$.
If $R$ is the current run, we have the following properties: 
(1) $R_{sim}$ only contains steps of participating processes of $R$, 
(2) $R_{sim}$ is $k$-concurrent, and 
(3) every participating \wf-process of $R$ that takes infinitely many steps  is given enough steps in $R_{sim}$ to decide.

\subsection{Simulating $k$ codes using $\Nomega_k$}
\label{app:kcodes}

Suppose we are given a read-write algorithm $\B$ on $k$ \wf-processes,
$p_1',\ldots,p_k'$. 
Assuming that $\Vomega_k$ is available, the algorithm in
Figure~\ref{fig:vomegabg} 
describes how $n$ \emph{simulators}, \wf-processes $p_1,\ldots,p_n$ can simulate
an infinite run of $\B$.

The simulation is similar in spirit to
BG-simulation~\cite{BG93b,BGLR01}.
Every simulator $p_i$ first registers its participation in the shared memory
and then tries to advance simulated \wf-processes $p_1',\ldots,p_{\min(k,m)}'$,
where $m$ is the number of simulators that $p_i$ has witnessed
participating.

To simulate a step of $p_j'$, simulators agree on the view of the
\wf-process after performing the step. 
However, instead of the BG-agreement protocol of~\cite{BG93b,BGLR01},  
we use here a \emph{leader-based} consensus algorithm~\cite{CHT96}.
In the algorithm, a process periodically (in every round $r$ of
computation), queries the current leader to
get an estimate of the decision.  

Since in our algorithm both \wf-processes and \fd-processes can be
elected leaders, we modify the algorithm of~\cite{CHT96} as
follows.
When a process wants to get an estimate of the decision (say in round
$r$), it publishes a query $(\id{query},\id{est}',r)$ in the shared
memory (proposing its current estimate $\id{est}'$), waits until the
current leader publishes a response  $(\id{est},r)$, and adopts the estimate.   
For simplicity, we assume that every process (be it a \wf-process or a
\fd-process) periodically scans the memory to find new queries of the
kind $(\id{query},\id{est}',r)$ and responds to them by publishing 
one of the proposed estimates. 
Furthermore, we assume that each \fd-process periodically updates 
the shared array $\Vomega_k\id{-\fd}[1,\ldots,k]$ with the output of its
module of  $\Vomega_k$.
Recall that eventually some position $\Vomega_k\id{-\fd}[j]$
($j\in\{1,\ldots,k\}$) stabilizes on the identity of some correct \fd-process.
 
The resulting algorithm terminates under the condition that all
\wf-processes eventually agree on the same correct leader.
The instance of the consensus algorithm used to simulate $\ell$-th
step of \wf-process $p_j'$ is denoted by $\cons_{j,\ell}$.

The rule to elect the leader is the following.
As long as the number of participating simulators is $k$ or less, the
participating simulator with the $j$-th smallest identity acts as a leader
for simulating steps of $p_j'$. 
When the number of participating simulators exceeds $k$, 
the leader for simulating steps of $p_j'$ is given by $\Vomega_k\id{-\fd}[j]$. 

In both cases, at least one simulated \wf-process is eventually associated with
the same correct leader. Thus, at least one
simulated \wf-process makes progress in the simulation.   

The algorithm also assumes that a simulator $p_i$ may decide to leave the
simulation if the simulated run produced a desired output
(line~\ref{line:depart}).
We use this option in the next section.

\begin{figure}[tbp]
\hrule \vspace{1mm} {\small
\begin{tabbing}
 bbb\=bb\=bb\=bb\=bb\=bb\=bb\=bb \=  \kill
Shared variables: \\
\> $R_j$, $j=1,\dots,m$, initially $\bot$\\
\> $V_j$, $j=1,\ldots,k$,  initially the initial state of $p_j'$\\
\> $\Vomega_k\id{-\fd}[j]$, $j=1,\dots,n$, initially $q_1$\\
Local variables: \\
\>\> $\Leader_j$, $j=1,\ldots,k$, initially $p_1$\\
\>\> $\ell_j$, $j=1,\ldots,k$, initially $1$\\
\>\> $v_j$, $j=1,\ldots,k$, initially $\bot$\\
\\
Task 1:\\
\nnll\label{line:registerd}\>  $R_i:=1$ \\
\nnll\> $\id{undecided} := \true$\\
\nnll\>  {\bf for} $j=1,\ldots,k$ {\bf do} $v_j := \{V_1,\ldots,V_k\}$\\
\nnll\> {\bf while} \id{undecided} {\bf do}\\
\nnll\label{line:parts}\>\> {\bf for} $j=1,\ldots,\min(|\id{pars}|,k)$ {\bf do}\\
\nnll\>\>\> perform one more step of $\cons_{j,\ell_j}(v_j)$ using $\Leader_j$
as a leader\\
\nnll\>\>\> {\bf if} $\cons_{j,\ell_j}(v_j)$ returns $v$  {\bf then}
\hspace{1cm} \{ The next state of $p_j'$ is decided \}\\ 
\nnll\>\>\>\>  $V_j := v$ \hspace{1cm} \{ Adopt the decided state of
$p_j'$ \}\\
\nnll\>\>\>\>  simulate the next step of $p_j'$ in $\B$\\
\nnll\label{line:candepart}\>\>\>\> {\bf if}  $v$ allows $p_i$ to decide  {\bf then}
\hspace{1cm} \{ The simulator can depart \}\\  
\nnll\>\>\>\>\>  $\id{undecided} := \false$\\
\nnll\label{line:depart}\>\>\>\>\>  $R_i := \bot$\\
\nnll\>\>\>\>  $v_j := \{V_1,\ldots,V_k\}$ \hspace{1cm}  \{ Evaluate
the next state of $p_j'$ \}\\
\nnll\>\>\>\> $\ell_j:=\ell_j+1$\\
\\
Task 2:\\
\nnll\label{line:leaders}\>   {\bf while} \id{true} {\bf do}\\
\nnll\>\> $\id{pars}:=\{p_j, R_j\neq\bot\}$\\
\nnll\>\> {\bf if} $|\id{pars}|\le k$ {\bf then}\\
\nnll\>\>\>   {\bf for} $j=1,\ldots,|\id{pars}|$ {\bf do} $\Leader_j :=$
the $j$-th smallest process in \id{pars}\\
\nnll\>\>  {\bf else}\\
\nnll\label{line:leaderf}\>\>\>   {\bf for} $j=1,\ldots,k$ {\bf do} $\Leader_j :=\Vomega_k\id{-\fd}[j]$
\end{tabbing}
\hrule 
}
\caption{Simulating $k$ codes using vector-$\Omega_k$: the program code 
  for simulator $p_i$}
\label{fig:vomegabg}
\end{figure}

\begin{theorem}
\label{th:kcodes}
In every environment, the protocol in Figure~\ref{fig:vomegabg} 
simulates an infinite run of any $k$-processes algorithm $\B$ 
(as long as there is at least one not decided participating simulated process).
Moreover, if $\ell$ simulators participate, i.e., $|\id{pars}|=\ell$, then
at most $\min(k,\ell)$ processes participate in the simulated run.
\end{theorem}
\begin{proof}
Consider an infinite  run of the algorithm.  
Since every next state of each simulated process $p_j'$ is decided
using a consensus algorithm, every simulator observes exactly the same
evolution of states for every simulated process. 
Thus, the simulated schedule indeed belongs to a run of $\B$.

Now consider the construction of variables $\Leader_1$, $\ldots$, $\Leader_k$ used by the
consensus algorithms $\cons_{1,\ell}$, $\ldots$, $\cons_{k,\ell}$ (lines~\ref{line:leaders}-\ref{line:leaderf}). 
Let $\ell$ be the number of participating simulators.

If $\ell\le k$, the simulator with the $j$-th
smallest identity in $\id{pars}$ is assigned to
be the leader of exactly one simulated process $p_j'$.
Since at least one simulator is correct,
there exists $p_j'$ ($j=1,\ldots,|\id{pars}|$)  such that all
instances $\cons_{j,\ell_j}$ using $Leader_j$ eventually terminate.  
Thus, $p_j'$ accepts infinitely many steps in the simulated run.

If $\ell>k$, at least one
$Leader_j$ ($j=1,\ldots,k$) eventually stabilizes on some correct process identity, 
as guaranteed by the properties of $\Vomega_k$.
Again, $p_j'$ takes infinitely  many steps in the simulated run.   

In both cases, at most $\min(\ell,k)$ simulated processes appear in the produced
run of $\B$, and at least one  simulated process takes infinitely many steps.
\end{proof}

\subsection{Solving a $k$-concurrent task with $\Nomega_k$}
\label{app:kconc}

\begin{result-repeat}{Theorem}{th:omega}
Let $T$ be any  $k$-concurrently solvable task.
In every environment~$\E$, $\Nomega_k$ solves $T$ in~$\E$.
\end{result-repeat}
\begin{proof}
Let $\A$ be the algorithm that solves $T$  $k$-concurrently.
We simply employ the simulation protocol in
Figure~\ref{fig:vomegabg}  (Theorem~\ref{th:kcodes}),
and suppose that the simulated algorithm $\B$ is
Extended BG-simulation~\cite{Gafni09} for $\A$.
More precisely, $\B$ simulates  with $k$ \wf-processes the algorithm $\A$ with  $\nbfinal$
\wf-processes.

Thus, the double simulation is built as follows.
Every process $p_i$ writes its input value of $T$ to the shared
memory and starts the simulation of $k$ processes $p_1',\ldots,p_k'$ 
using the algorithm in Figure~\ref{fig:vomegabg}.
The simulated processes $p_1',\ldots,p_k'$ run, in turn, 
BG-simulation of $\A$ on $\nbfinal$ processes $p_1'',\ldots,p_\nbfinal''$.

Each simulated process $p_j''$ is simulated only if the corresponding
$p_j$ has written its input of $T$ in the shared memory
and $p_j''$ has not yet obtained an output in the simulated run.
Moreover, to make sure that the simulation indeed produces a $k$-concurrent
run, at any point of the simulation, each simulator in $p_j'\in \{p_1',\ldots,p_k'\}$
tries to advance the participating and not yet decided process with
the smallest id.     
If the currently simulated process is found
blocked~\cite{BG93b,BGLR01}, 
i.e., the process cannot advance because another simulator started simulating a step
of it but has not yet finished, $p_j'$ proceeds to the next smallest
undecided participating process in $\{p_1'',\ldots,p_n''\}$.
Since there are at most $k$ simulators, at most $k-1$ undecided participating  processes
can be found blocked and thus there are at most $k$ undecided participating processes at a
time---the resulting simulated run is $k$-concurrent. 

When $p_j''$ obtains an output, the corresponding simulator $p_j$
considers itself ``decided'' (line~\ref{line:candepart}),
writes $\bot$ in $R_i$ (line~\ref{line:depart}) and departs. 

If $p_i'$ cannot make progress because each code it tries to simulate 
is blocked and there are no more codes to add, it ``aborts'' 
all blocked agreements~\cite{Gafni09} and resumes the simulation. 
Since, at each point of time, the number of simulated  codes 
does not get below the number of simulators that take steps, the simulation keeps making progress.
 
Thus, as long as $\ell$ processes $\{p_{j_1},\ldots,p_{j_\ell}\}$ participate, 
only $\min(k,\ell)$ processes in $\{p_{1}'',\ldots,p_{n}''\}$ take
steps, which results in a $k$-concurrent simulated run of $\A$.
Every process $p_j''$ that takes steps eventually decides in a $k$-concurrent run of $\A$ 
and the corresponding simulator $p_j$ departs.
As soon as the decided process $p_i$ departs by writing $\bot$ to $R_i$, we
have one simulator $p_i$ and one simulated process $p_i''$ less.
Therefore, as long as there is a simulator taking steps and the run is fair, the simulated run makes
progress, i.e., more and more participants decide.
Thus, we obtain an algorithm that, in every environment,
solves $T$.  
\end{proof}

\section{Proof for characterizing the task of renaming (Section~\ref{sec:renaming})}
\label{app:renaming}

To illustrate the utility of our  framework, we consider the task of \emph{$(j,\ell)$-renaming}~\cite{ABDPR90}.
The task is defined on $\nbfinal$ ($\nbfinal>j$) processes and assumes that in every 
run at most $j$ processes participate (at least $\nbfinal-j$ elements of each vector $I\in\I$ are $\bot$). 
As an ouput, every participant obtains a unique \emph{name} in the range $\{1,\ldots,\ell\}$ 
(every non-$\bot$ element in each $O\in\O$ is a distinct value in $\{1,\ldots,\ell\}$).

We show first that $(j,j)$-renaming (also called \emph{strong} $j$-renaming) is not $2$-concurrently solvable.
Then we present a generic algorithm that, for all $k=1,\ldots,j$,
solves $(j,j+k-1)$-renaming in all $k$-concurrent run,
and thus $(j,j+k-1)$-renaming can be solved (in IFD) using $\Nomega_k$.

\subsection{Impossibility of $2$-concurrent strong $2$-renaming}

\begin{result-repeat}{Lemma}{lem:sr-imp}
Strong $2$-renaming  cannot be solved $2$-concurrently.
\end{result-repeat}
\begin{proof}
We start with showing that for the special case of $j=2$, strong renaming cannot be solved $2$-concurrently.
Suppose, by contradiction, that there exists a (restricted) algorithm $\A$ that solves $(2,2)$-renaming $2$-concurrently.
Since we assumed $j<n$, we have at least $3$ processes in the system.
By the pigeon-hole principle, there exist two processes that decide on the same name $v\in\{1,2\}$ in their solo runs of $\A$.   Without loss of generality, let these processes be $p_1$ and $p_2$ and let $v$ be $1$. 

Now $p1$ and $p2$ can wait-free solve $2$-processes consensus as follows.
Each process publishes its input and then runs $\A$ until it obtains a name. 
If the name is $1$, the process decides on its input, otherwise it decides on the input of the other process.
Since a process in $\{p_1,p_2\}$ obtains $1$ as a name in a solo run of $\A$, if $1$ is not obtained, then 
the other process participates in the run of $\A$ and, thus, has previously written its input. 
Therefore, every decided value was previously proposed. 
Since every obtained name is distinct, the two processes cannot decide on different values.
This conclude the proof that strong 2-renaming cannot be 2-concurrently solvable.
\end{proof}

\begin{figure}[tbp]
\hrule \vspace{1mm} {\small
\begin{tabbing}
 bbb\=bbB\=bbb\=bbb\=bbb\=bbb\=bbb\=bbb\=bbb\=bbb\=bbb\=bbb\=bbb \= \kill
Shared variables: $R_\ell$, $\ell=1,\ldots,\nbfinal$, initially $\bot$\\ 
\nnll\label{line:regpar}\> $R_i := 1$ \>\>\>\>\>\>\>\>\>\> /* \texttt{register participation} */\\
\nnll\> \textbf{repeat}\\
\nnll\label{line:parset}\>\> $S:=\{p_\ell\;|\; R_\ell\neq\bot\}$ \>\>\>\>\>\>\>\>\> /*\texttt{get the current participating set} */\\
\nnll\label{line:notdecidedset}\>\> $S':=\{p_\ell\;|\; R_\ell=1\}$ \>\>\>\>\>\>\>\>\>
/*\texttt{get the set of not yet decided participants} */\\
\nnll\>\> $min_1:=\min(S')$\\
\nnll\>\> \textbf{if} ($|S'|=1$ ) \textbf{then} $min_2:=min_1$  \textbf{else} 
$min_2:=\min(S' -\min(S'))$\\
\nnll\label{line:checkifmin}\>\> \textbf{if} ($|S|=j$ and ($p_i=min_1$ or $p_i=min_2$))
or ($|S|=j-1$ and $p_i=min_1$) \textbf{then}\\
\nnll\>\>\>\label{line:takestepA} take one more step of $\A$
\>\>\>\>\>\>\>\> /*\texttt{if among two not decided with smallest ids} */\\
\nnll\> \textbf{until} decided\\
\nnll\label{line:decdep}\> $R_i := 0$\\
\nnll\> return the name decided in $\A$
\end{tabbing}
\vspace{-2mm}\hrule }
\caption{\small A $1$-resilient strong $j$-renaming algorithm: code
  for each $\wf$-process $p_i$.}
\label{fig:resrenaming}
\end{figure}

\begin{result-repeat}{Theorem}{th:sr-imp}
For all $1<j<\nbfinal$,  strong $j$-renaming  cannot be solved $2$-concurrently.
\end{result-repeat}   
\begin{proof}
By Lemma~\ref{lem:sr-imp}, we have already the result for $j=2$. 
Suppose, by contradiction, that for some $2<j<\nbfinal$, there exists an (restricted) algorithm $\A$ solving strong $j$-renaming $2$-concurrently.
As we deal here with $2$-concurrent solvability, we are only interested by the $\wf$-processes and their algorithms.
We use $\A$ to solve strong $j$-renaming in all \emph{$1$-resilient}
runs, i.e.,  runs in which at least $j-1$ $\wf$-processes
participate and take infinitely many steps. Recall that at most $j$
$\wf$-processes participate in every run, so either $j-1$ or $j$ processes
take infinitely many steps.
In the algorithm (Figure~\ref{fig:resrenaming}), every process
registers its participation (line~\ref{line:regpar}) and then
periodically checks the current set of participants (line~\ref{line:parset}). 
If it finds out that it is among $2$ processes with the smallest identities
among $j$ participating but not yet processes
(line~\ref{line:checkifmin}), 
then it starts taking steps $\A$ until the algorithm provides $p_i$
with a new name. 
Then $p_i$ declares that it has decided (line~\ref{line:decdep}) and departs. 
    
Note that the resulting run of $\A$ is $2$-concurrent: either
the participating set is of size $j-1$ and only the not yet decided
participant with the smallest identity is allowed to take steps of $\A$ solo, or  
exactly $j$ processes participate and the two not yet
decided processes with the smallest identity are allowed to take step concurrently. 

Now we observe that the run of $\A$ continues as long as there
is at least one not yet decided participant that take steps. 
Indeed, either the participating set is of size $j-1$ and every
participant takes an infinity number of steps (including the not yet decided one with the
smallest identity) 
or exactly $j$ \wf-processes participate and at least one of the not yet
decided processes with the two smallest identity takes an infinity number of steps.     
Thus, every \wf-process that keeps taking steps of $\A$ in the resulting 
$2$-concurrent run eventually decides and departs.
The set of undecided participants gets smaller by one, and the
next \wf-process with the smallest identity joins the $2$-concurrent run
of $\A$. 

But it is shown in~\cite{Gafni09} that if all $1$-resilient runs of a restricted algorithm $\A$ satisfy
strong $j$-renaming then there is a restricted algorithm to solve strong $2$-renaming 
$2$-concurrently---a contradiction with Lemma~\ref{lem:sr-imp}. 
\end{proof}

\subsection{Solving renaming}
\label{app:solvingrenaming}

The  distributed algorithm used to solve  $(j,j+k-1)$-renaming    $k$-concurrently essentially mimics the algorithm of~\cite{ABDPR90,AW04} for wait-free 
$(j,2j-1)$-renaming.

\begin{theorem}
\label{th:kconcrenaming}
For all $1<k\leq j<m$, $(j,j+k-1)$-renaming  can be solved $k$-concurrently.
\end{theorem}   
\begin{proof}
Our algorithm, described in Figure~\ref{fig:kconcrenaming}, essentially 
mimics the algorithm of~\cite{ABDPR90,AW04} for wait-free 
$(j,2j-1)$-renaming.
\begin{figure}[tbp]
\hrule \vspace{1mm} {\small
\begin{tabbing}
 bbb\=bbB\=bbb\=bbb\=bbb\=bbb\=bbb\=bbb\=bbb\=bbb\=bbb\=bbb\=bbb \= \kill
Shared variables:\\
\> $R_\ell$, $\ell=1,\ldots,\nbfinal$, initially $\bot$\\ 
\\   
\nnll\label{line:init}\> $s := 1$ \\
\nnll\> \textbf{repeat forever}\\
\nnll\label{line:updatename}\>\> $R_i:=(i,s,\true)$ \>\>\>\>\>\>\>\>\>\>
/*\texttt{register new name}*/\\
\nnll\label{line:scannames}\>\>
$S:=\{p_\ell\;|\; R_\ell\neq\bot\}$ 
 \>\>\>\>\>\>\>\>\>\>
/*\texttt{collect suggested names}*/\\ 
\nnll\label{line:checkifalone}\>\> \textbf{if} $\exists (\ell,s_\ell,b)\in S$:
$i\neq \ell$ and $s=s_\ell$ \textbf{then}\\
\nnll\label{line:rank}\>\>\> $r := $ the rank of $i$ in
$\{\ell\:|\:(\ell,s_\ell,b)\in S, b=\true\}$\\ 
\>\>\>\>\>\>\>\>\>\>\>\> /*\texttt{rank among not yet decided participants}*/\\
\nnll\label{line:newname}\>\>\> $s := $ the $r$th integer not in
$\{s_\ell:|\:(\ell,s_\ell,b)\in S, i\neq \ell\}$\\
\>\>\>\>\>\>\>\>\>\>\>\> /*\texttt{suggest a new name among not yet suggested}*/\\
\nnll\>\> \textbf{else}\\
\nnll\>\>\> $R_i:=(i,s,\false)$\\ 
\nnll\>\>\> return $s$
\end{tabbing}
\vspace{-2mm}\hrule }
\caption{\small A $k$-concurrent $(j,j+k-1)$-renaming algorithm: code
  for each process $p_i$.}
\label{fig:kconcrenaming}
\end{figure}
  
In the algorithm, every process periodically selects a new name
according to the set of the names not yet suggested by other processes
and its rank among the set of currently not yet decided participants
(lines~\ref{line:rank} and~\ref{line:newname}).    

Note that since at most $j$ processes participate in every run, 
$p_i$ can observe at most $j-1$ names suggested by
other processes in line~\ref{line:scannames}. 
Furthermore, since in a $k$-concurrent run, $p_i$ can observe at
most $k$ not yet decided participants, its rank can be at most $k$.
Therefore, the highest name $p_i$ can suggest in
line~\ref{line:updatename} is $j+k-1$.

Now we show that no two processes output the same name.
Suppose, by contradiction, that $p_i$ and $p_j$ output the same name
$s$. Thus, both $p_i$ and $p_j$ previously suggested $s$ in
line~\ref{line:updatename}.
But since after than both processes read each other's registers after
that, at least one of them would see that $s$ has been suggested by
another process and thus would not decide---a contradiction.  

Finally,  we show that  every correct process eventually decides. 
Consider, by contradiction, an run $R$ in which a set of correct processes
$\{p_{j_1},\ldots,p_{j_t}\}$ (ordered by their ids) never decide.
We call these processes \emph{trying}.
We establish a contradiction by showing that $p_{j_1}$ must eventually
decide. Indeed, consider $R'$, a prefix of  $R$, in which only trying
processes take steps, and let $S$ be the set of names suggested 
by the processes not in $\{p_{j_1},\ldots,p_{j_t}\}$ (note that this set does
not change in $R$). 
Since, $p_{j_1}$ has the smallest rank among the trying processes (let us
denote it by $r$), eventually no trying process will ever suggest
the $r$th name not in $S$. Thus, $p_{j_i}$ eventually finds itself
to be the only process to suggest the name and decides---a contradiction.
\end{proof} 
From this result and Theorem~\ref{th:omega}, we can conclude:
\begin{theorem}
\label{appth:kconcren}
For all $1<k\leq j<m$, $(j,j+k-1)$-renaming  can be solved with $\Nomega_k$.
\end{theorem}

\end{document}